\documentclass{svproc}
\usepackage{url}

\usepackage{amssymb}
\usepackage{amsmath}
\usepackage{braket}
\usepackage{graphicx}

\begin{document}
\mainmatter              
\title{Trace decreasing quantum dynamical maps: Divisibility and entanglement dynamics}
\titlerunning{Trace decreasing quantum dynamical maps}  
\toctitle{Trace decreasing quantum dynamical maps: Divisibility
and entanglement dynamics}
\author{Sergey N. Filippov\inst{1,2}}
\authorrunning{Sergey N. Filippov} 
%
\tocauthor{Sergey N. Filippov}
\institute{Steklov Mathematical Institute of Russian Academy of
Sciences, Gubkina Street 8, Moscow 119991, Russia \and Valiev
Institute of Physics and Technology of Russian Academy of
Sciences, Nakhimovskii Prospect 34, Moscow 117218, Russia}

\maketitle              

\begin{abstract}
Trace decreasing quantum operations naturally emerge in
experiments involving postselection. However, the experiments
usually focus on dynamics of the conditional output states as if
the dynamics were trace preserving. Here we show that this
approach leads to incorrect conclusions about the dynamics
divisibility, namely, one can observe an increase in the trace
distance or the system-ancilla entanglement although the trace
decreasing dynamics is completely positive divisible. We propose
solutions to that problem and introduce proper indicators of the
information backflow and the indivisibility. We also review a
recently introduced concept of the generalized erasure dynamics
that includes more experimental data in the dynamics description.
The ideas are illustrated by explicit physical examples of
polarization dependent losses.
\end{abstract}
\section{Introduction}
Many quantum physics experiments involve a postselection, i.e., an
analysis of only those events that are conditioned by some
measurement outcomes. Examples include the approximate preparation
of a single-photon state via detecting a heralding photon in the
parametric down-conversion process~\cite{uren-2004}, the
implementation of the conditional quantum
gates~\cite{kiesel-2005}, the generation of high-dimensional
maximally entangled orbital-angular-momentum
states~\cite{kovlakov-2018}, the conditional subtraction of a
single photon~\cite{wenger-2004} or multiple
photons~\cite{bogdanov-2017} from the electromagnetic field, the
conditional addition of a single photon to the electromagnetic
field~\cite{zavatta-2004}, the experimental evaluation of
so-called weak values~\cite{pryde-2005}, and the simulation of
quantum collision models~\cite{mataloni-2019}. A proper
mathematical description of any postselection procedure is given
by a quantum operation, i.e., a trace nonincreasing completely
positive map $\Lambda$ (see, e.g.,~\cite{kraus-1983}) that is
associated with a specific measurement outcome. If $\varrho$ is an
initial density operator before postselection, then
$\Lambda[\varrho]$ is a subnormalized density operator such that
the trace ${\rm tr}\big[ \Lambda[\varrho] \big]$ is exactly the
probability to get that specific measurement outcome.

In general, a measurement-induced transformation of the system
state is fully described in terms of the quantum instrument that
assigns a quantum operation to each measurement
outcome~\cite{davies-lewis-1970,holevo-2012,heinosaari-ziman}. A
quantum operation $\Lambda$ corresponds to a specific measurement
outcome, say, a successful detection of the heralding photon in
quantum optics
experiments~\cite{uren-2004,kiesel-2005,wenger-2004,zavatta-2004}.
In mathematical description of experiments with photon subtraction
or photon addition~\cite{wenger-2004,zavatta-2004}, however, the
conventionally used transformations $\varrho \rightarrow t a
\varrho a^{\dag}$ and $\varrho \rightarrow t a^{\dag} \varrho a$
are not legitimate quantum operations for any constant $t>0$ due
to the violation of the trace nonincreasing
property~\cite{filippov-ljm-2019} ($a$ and $a^{\dag}$ are the
photon annihilation operator and the photon creation operator for
a considered mode, respectively). The transformations above
approximate the faithful quantum operations of photon substraction
and addition for density operators from some
class~\cite{filippov-ljm-2019}.

To illustrate a simple trace nonincreasing operation, let us
consider an experiment with polarization qubits. Suppose $\varrho$
is a density operator for polarization degrees of freedom of
single photons and those single photons propagate through a lossy
optical fiber with the intensity attenuation factor $p$, then
$\Lambda[\varrho] = p \varrho$ and ${\rm tr}\big[ \Lambda[\varrho]
\big] = p$. The attenuation factor is experimentally estimated as
$N_{\rm out}/N_{\rm in}$, where $N_{\rm out}$ is the number of
photons detected at the output of the fiber and $N_{\rm in}$ is
the number of photons that enter the fiber. If the photon detector
has a finite quantum efficiency $\eta$, then the number of
successfully detected photons at the input (the output) of the
fiber approximately equals $N_{\rm in}' = \eta N_{\rm in}$
($N_{\rm out}' = \eta N_{\rm out}$) so that $N_{\rm out}' / N_{\rm
in}' = N_{\rm out}/N_{\rm in}$, i.e., ${\rm tr}\big[
\Lambda[\varrho] \big]$ can be correctly estimated even with the
use of imperfect devices. Alternatively, the photon pairs can be
produced in the parametric down-conversion process, with one
photon being sent to a lossy channel (to count $N_{\rm out}'$ at
the output) and the other one being used as a herald (to count
$N_{\rm in}'$). These ideas enable one to reconstruct a trace
nonincreasing operation based on experimental data. For instance,
the experimental quantum process tomography of a trace decreasing
map $\Lambda$ describing the partially transmitting polarizing
beam splitter is reported in Ref.~\cite{mataloni-2010}.

The whole idea of postselection is focused on the conditional
output state
\begin{equation*}
\varrho_{\rm out} = \frac{\Lambda[\varrho]}{{\rm tr}\big[
\Lambda[\varrho] \big]} \equiv \Lambda_{{\cal D}}[\varrho]
\end{equation*}

\noindent that is a valid density operator provided ${\rm tr}\big[
\Lambda[\varrho] \big] > 0$. (If ${\rm tr}\big[ \Lambda[\varrho]
\big] = 0$, then there is no sense in postselection because a
desired measurement outcome is never observed.) The introduced map
$\Lambda_{{\cal D}}$ defines a nonlinear transformation of density
operators that finds applications in the stroboscopic
implementation of effective non-Hermitian
Hamiltonians~\cite{luchnikov-2017,grimaudo-2020}. In this work, we
are primarily interested in biased operations $\Lambda$ that have
the property ${\rm tr} \big[ \Lambda[\varrho_1] \big] \neq {\rm
tr} \big[ \Lambda[\varrho_2] \big]$ for at least two density
operators $\varrho_1$ and $\varrho_2$~\cite{filippov-2021}.
Examples of biased quantum operations include the partially
transmitting polarizing beam splitter~\cite{mataloni-2010} and the
polarization-dependent losses~\cite{gisin-1997}.

Interestingly, the optical simulator of quantum
collisions~\cite{mataloni-2019} produces a sequence of quantum
operations $\Lambda(t_m)$ in the single-photon sector that
correspond to transformations from time $t=0$ to a number of
different time moments $t=t_m$, $m=1,2,\ldots$. Similarly, one can
consider losses in optical fibers of different length to get a
one-parameter family of quantum operations $\{\Lambda(t)\}_{t \geq
0}$ that we refer to as a trace decreasing dynamical map if ${\rm
tr}\big[ \Lambda(t)[\varrho] \big] < 1$ for some density operator
$\varrho$ and some $t \geq 0$. Note that the relation ${\rm tr}
\big[ \Lambda(t_1)[\varrho] \big] \geq {\rm tr} \big[
\Lambda(t_2)[\varrho] \big]$ for $t_2 \geq t_1 \geq 0$ need not
hold in general. Suppose $\Lambda(t_2) = \Theta(t_2,t_1)
\Lambda(t_1)$ for all $t_2 \geq t_1 \geq 0$, where
$\Theta(t_2,t_1)$ is a legitimate quantum operation (i.e., a
completely positive and trace nonincreasing map). Then
$\Lambda(t)$ is called completely positive divisible
(CP-divisible). Physical meaning of CP-divisibility is that the
dynamics $\Lambda(t)$ can be effectively viewed as a concatenation
of independent subevolutions.

As far as trace preserving dynamical maps $\{\Phi(t)\}_{t \geq 0}$
are concerned, CP-divisibility is studied in a number of papers
and many different indicators of CP-indivisibility are proposed
too (see,
e.g.,~\cite{rivas-2014,chruscinski-2018,rivas-2010,chruscinski-2017,fc-2018,blp-2009}).
However, in the following Secs. we show that these indicators
become misleading if carelessly applied to postselected states of
a trace decreasing dynamical map $\{\Lambda(t)\}_{t \geq 0}$. Note
that some authors associate CP-divisibility with a quantum version
of the Markovian process~\cite{rivas-2014}; however, a more
complicated relation takes place operationally~\cite{milz-2019}.

\section{Trace distance approach to non-Markovianity}
\label{section-distinguish}

Let $\varrho_1$ and $\varrho_2$ be density operators, then the
trace distance $\frac{1}{2} \| \Phi(t) [\varrho_1] - \Phi(t)
[\varrho_2] \|_1$ is a nonincreasing function of time $t$ if the
trace preserving dynamical map $\{\Phi(t)\}_{t \geq 0}$ is
CP-divisible~\cite{blp-2009}. The physical meaning of the trace
distance is related with the maximum success probability to
distinguish the initially equiprobable states $\varrho_1$ and
$\varrho_2$ after time $t$. The maximum success probability equals
$\frac{1}{2} + \frac{1}{4} \| \Phi(t) [\varrho_1] - \Phi(t)
[\varrho_2] \|_1$ in this case (see,
e.g.,~\cite{heinosaari-ziman}). If the trace distance diminishes,
then the success probability diminishes too, which is treated as a
flow of information from the open system to its environment. On
the other hand, the increase of the trace distance is an
indication of the backflow of information according to
Ref.~\cite{blp-2009}.

If one experimentally reconstructs the conditional output states
$\Lambda_{\cal D}(t)[\varrho_1] = \Lambda(t)[\varrho_1] / {\rm
tr}\big[ \Lambda(t) [\varrho_1] \big]$ and $\Lambda_{\cal
D}(t)[\varrho_2] = \Lambda(t)[\varrho_2] / {\rm tr}\big[
\Lambda(t) [\varrho_2] \big]$, then it is tempting to use the
conventional trace distance
\begin{equation} \label{td-naive}
\frac{1}{2} \| \Lambda_{\cal D}(t)[\varrho_1] - \Lambda_{\cal
D}(t)[\varrho_2] \|_1
\end{equation}

\noindent to analyze the information flow. However, this approach
encounters a number of problems. First, the probabilities
$p_1(t):={\rm tr}\big[ \Lambda(t) [\varrho_1] \big]$ and
$p_2(t):={\rm tr}\big[ \Lambda(t) [\varrho_2] \big]$ differ in
general. This means that by using Eq.~\eqref{td-naive} an
experimentalist disregards some extra information on how often the
states $\Lambda_{\cal D}(t)[\varrho_1]$ and $\Lambda_{\cal
D}(t)[\varrho_2]$ are actually produced. Second, as we show in the
example below, the quantity~\eqref{td-naive} may increase even for
a semigroup trace decreasing dynamics $\Lambda(t)$ that is
obviously CP-divisible.

\begin{example} \label{example-distinguishability-naive}
Consider the polarization qubit dynamics $\Lambda(t) = e^{L t}$,
where
\begin{equation} \label{L-example}
L[\varrho] = - \frac{1}{2} \Big\{ \gamma_H \ket{H}\bra{H} +
\gamma_V \ket{V}\bra{V},\varrho \Big\},
\end{equation}

\noindent $\{ \cdot, \cdot \}$ stands for the anticommutator;
$(\ket{H},\ket{V})$ is a conventional basis composed of the
horizontally and vertically polarized states; $\gamma_H$ and
$\gamma_V$ are the attenuation rates for horizontally and
vertically polarized photons, respectively. Suppose $\varrho_1 =
\ket{H}\bra{H}$ and $\varrho_2 =
\frac{1}{2}(\ket{H}+\ket{V})(\bra{H}+\bra{V})$, then we explicitly
find the subnormalized states in the basis $(\ket{H},\ket{V})$,
\begin{equation*}
\Lambda(t)[\varrho_1] = \left(%
\begin{array}{cc}
  e^{-\gamma_H t} & 0 \\
  0 & 0 \\
\end{array}%
\right), \quad \Lambda(t)[\varrho_2] = \frac{1}{2} \left(%
\begin{array}{cc}
  e^{-\gamma_H t} & e^{-(\gamma_H + \gamma_V) t / 2} \\
  e^{-(\gamma_H + \gamma_V) t / 2} & e^{-\gamma_V t} \\
\end{array}%
\right),
\end{equation*}

\noindent and the corresponding conditional output states,
\begin{equation*}
\Lambda_{\cal D}(t)[\varrho_1] = \left(%
\begin{array}{cc}
  1 & 0 \\
  0 & 0 \\
\end{array}%
\right), \ \Lambda_{\cal D}(t)[\varrho_2] = \frac{1}{e^{-\gamma_H t} + e^{-\gamma_V t}} \left(%
\begin{array}{cc}
  e^{-\gamma_H t} & e^{-(\gamma_H + \gamma_V) t / 2} \\
  e^{-(\gamma_H + \gamma_V) t / 2} & e^{-\gamma_V t} \\
\end{array}%
\right).
\end{equation*}

\noindent Therefore, $\frac{1}{2} \| \Lambda_{\cal
D}(t)[\varrho_1] - \Lambda_{\cal D}(t)[\varrho_2] \|_1 = \left[ 1
+ e^{(\gamma_V-\gamma_H)t} \right]^{-1/2}$, which monotonically
increases with time $t$ if $\gamma_H > \gamma_V$. \hfill
$\triangle$
\end{example}

Example~\ref{example-distinguishability-naive} shows that
Eq.~\eqref{td-naive} cannot be used in quantification of the
information flow for trace decreasing dynamical maps. An
alternative approach is to utilize the conditional probabilities
$\frac{p_1(t)}{p_1(t) + p_2(t)}$ and $\frac{p_2(t)}{p_1(t) +
p_2(t)}$ to get the postselected state $\Lambda_{\cal
D}(t)[\varrho_1]$ and $\Lambda_{\cal D}(t)[\varrho_2]$,
respectively, provided the desired quantum operation $\Lambda(t)$
is successfully fulfilled (for either of the input states,
$\varrho = \varrho_1$ or $\varrho=\varrho_2$). Then the modified
trace distance
\begin{equation} \label{td-naive-2}
 \left\| \frac{p_1(t)}{p_1(t)+p_2(t)} \Lambda_{\cal D}(t)[\varrho_1] -
\frac{p_2(t)}{p_1(t)+p_2(t)} \Lambda_{\cal D}(t)[\varrho_2]
\right\|_1
\end{equation}

\noindent is a relevant candidate to track the information flow.
However, as we reveal in the example below, the
quantity~\eqref{td-naive-2} is still unsatisfactory because it
does not monotonically decrease under the semigroup dynamics.

\begin{example} \label{example-naive-2}
Let $\Lambda(t) = e^{L t}$, where $L$ is given by
Eq.~\eqref{L-example}. Suppose $\varrho_1 = \ket{H}\bra{H}$ and
$\varrho_2 = \frac{1}{2}(\ket{H}+\ket{V})(\bra{H}+\bra{V})$, then
\begin{equation*}
\frac{\| p_1(t) \Lambda_{\cal D}(t)[\varrho_1] - p_2(t)
\Lambda_{\cal D}(t)[\varrho_2] \|_1}{p_1(t)+p_2(t)}  = \sqrt{ 1 -
\dfrac{8 e^{-2 \gamma_H t}}{(3 e^{- \gamma_H t} + e^{- \gamma_V
t})^2} } \, ,
\end{equation*}

\noindent which increases with the increase of time $t$ if
$\gamma_H > \gamma_V$. \hfill $\triangle$
\end{example}

The solution of the problem is to multiply the conditional success
probability of the discrimination task,
\begin{equation}
\frac{1}{2} \left\{ 1 + \frac{\| p_1(t) \Lambda_{\cal
D}(t)[\varrho_1] - p_2(t) \Lambda_{\cal D}(t)[\varrho_2]
\|_1}{p_1(t)+p_2(t)} \right\},
\end{equation}

\noindent by the average probability to implement the operation
$\Lambda(t)$, i.e., $\frac{1}{2}[p_1(t) + p_2(t)]$. Then we get
$\frac{1}{4} \left\{ p_1(t)+p_2(t) + \| p_1(t) \Lambda_{\cal
D}(t)[\varrho_1] - p_2(t) \Lambda_{\cal D}(t)[\varrho_2] \|_1
\right\}$. Note that $p_i(t) \Lambda_{\cal D}(t)[\varrho_i] =
\Lambda(t)[\varrho_i]$, $i=1,2$, so the obtained expression
reduces to
\begin{equation} \label{success-correct}
p_{\rm succ.dist.} = \frac{1}{4} \bigg( {\rm tr}\big[
\Lambda(t)[\varrho_1] \big] + {\rm tr}\big[ \Lambda(t)[\varrho_2]
\big] + \| \Lambda(t) [\varrho_1] - \Lambda(t) [\varrho_2] \|_1
\bigg),
\end{equation}

\noindent which is a nonincreasing function of time $t$ if
$\Lambda(t)$ is CP-divisible. Indeed, following the lines of
Ref.~\cite{ruskai-1994}, we readily see that the trace distance
$\frac{1}{2} \| \Lambda(t) [\varrho_1] - \Lambda(t) [\varrho_2]
\|_1$ remains a nonincreasing function of time $t$ if we extend
the result of Theorem 1 in Ref.~\cite{ruskai-1994} to the case of
trace decreasing completely positive maps. If $\Lambda(t)$ is CP
divisible, then the probabilities ${\rm tr}\big[
\Lambda(t)[\varrho_1] \big]$ and ${\rm tr}\big[
\Lambda(t)[\varrho_2] \big]$ to successfully implement the
operation $\Lambda(t)$ for the input states $\varrho_1$ and
$\varrho_2$, respectively, are nonincreasing functions of time $t$
too. Therefore, the information flow in a trace decreasing quantum
dynamics is to be associated with the change in
Eq.~\eqref{success-correct}, whereas the naive expressions
\eqref{td-naive} and \eqref{td-naive-2} should be avoided.

\begin{figure}
\begin{center}
\includegraphics[width=10cm]{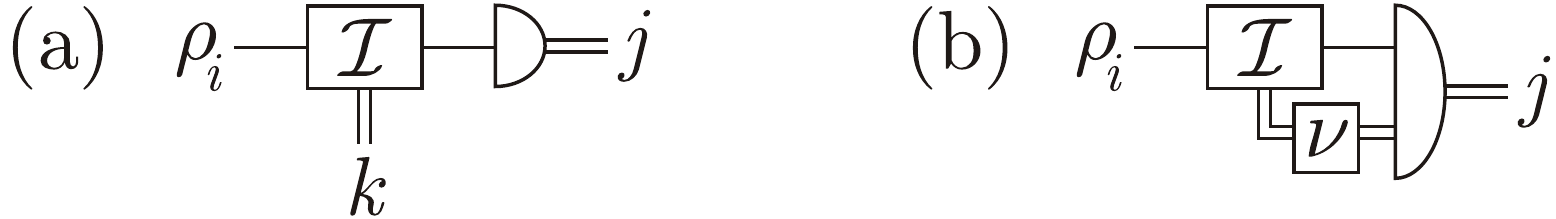}
\end{center}
\caption{{\bf (a)} Operational meaning of
Eq.~\eqref{success-correct}, where $\Lambda(t) = {\cal
I}_{k_{\ast}}$. {\bf (b)} Discrimination problem in the case of
the generalized erasure dynamics.} \label{figure-1}
\end{figure}

The operational meaning of Eq.~\eqref{success-correct} is depicted
in Fig.~\ref{figure-1}(a). Let $p(j,k|i)$ be a joint probability
to implement an operation ${\cal I}_k$ via a quantum instrument
${\cal I}$ and to observe the outcome $j$ while measuring the
output quantum state, provided the input state is $\varrho_i$.
Suppose ${\cal I}_{k_{\ast}} = \Lambda(t)$, then the maximum value
of $\frac{1}{2} \sum_{i=1,2} p(j=i,k=k_{\ast}|i)$ exactly equals
$p_{\rm succ.dist.}$ in Eq.~\eqref{success-correct}.

\section{System-ancilla entanglement dynamics}
\label{section-s-a-ent}

Suppose that in addition to the system of interest we also have
access to an ancillary system, which is isolated from the
decoherence sources. Provided the system dynamics is described by
a trace preserving dynamical map $\Phi(t)$, the total
system-ancilla aggregate undergoes an evolution given by the
dynamical map $\Phi(t) \otimes {\rm Id}$, see
Fig.~\ref{figure-2}(a).

\begin{figure}
\begin{center}
\includegraphics[width=10cm]{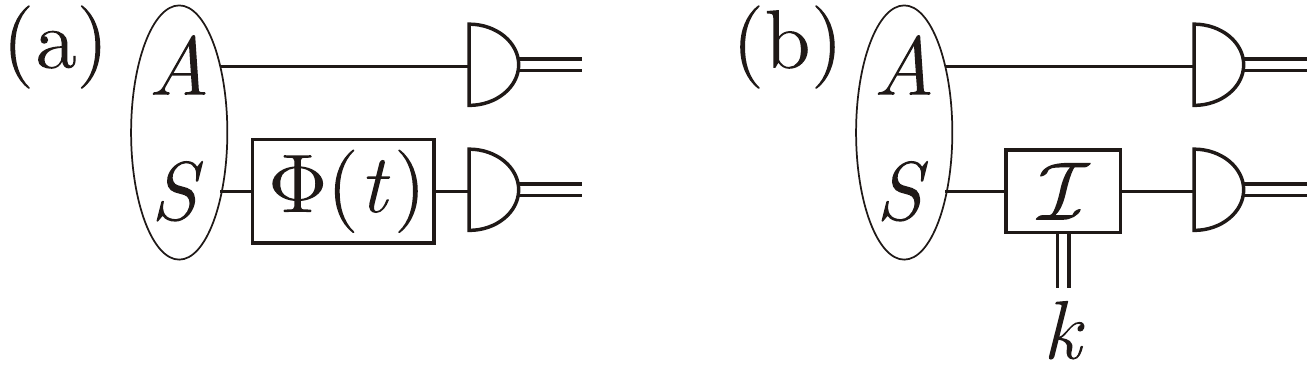}
\end{center}
\caption{{\bf (a)} Joint trace preserving dynamics of the system
($S$) and the ancilla ($A$). {\bf (b)} Schematic of an experiment
with trace decreasing dynamics, ${\cal I}_{k_{\ast}} =
\Lambda(t)$.} \label{figure-2}
\end{figure}

Suppose that the system and the ancilla are initially entangled,
then the system-ancilla entanglement generally changes in time. To
quantify the system-ancilla entanglement, we use some entanglement
monotone ${\cal E}$ that does not increase under trace preserving
local operations and classical communication~\cite{plenio-2007}. A
seminal observation of Ref.~\cite{rivas-2010} is that the
entanglement monotone ${\cal E}(t)$ is a nonincreasing function of
time $t$ provided the trace preserving map $\Phi(t)$ is
CP-divisible. Any increase in ${\cal E}(t)$ is an indication of
CP-indivisibility (non-Markovianity according to
Ref.~\cite{rivas-2014}). Although this approach to
CP-indivisibility detection is absolutely legitimate in the case
of trace preserving dynamical maps, it turns out to be wrong if
one tracks the entanglement dynamics of postselected
system-ancilla states [see setup in Fig.~\ref{figure-2}(b) with
postselection $k=k_{\ast}$]. The following example justifies this
claim.

\begin{example} \label{example-pdl-t}
Consider the polarization dependent losses that are time dependent
too, i.e., the master equation of the form
\begin{equation} \label{master-pdl}
\frac{d \varrho(t)}{dt} = - \frac{1}{2} \Big\{ \gamma_H(t)
\ket{H}\bra{H} + \gamma_V(t) \ket{V}\bra{V},\varrho(t) \Big\}.
\end{equation}

\noindent A solution of this master equation is
\begin{equation*}
\Lambda(t)[\varrho(0)] \equiv \varrho(t) = \left(%
\begin{array}{cc}
  e^{-\Gamma_H(t)} \varrho_{HH}(0) & e^{-\frac{1}{2}[\Gamma_H(t)+\Gamma_V(t)]} \varrho_{HV}(0) \\
  e^{-\frac{1}{2}[\Gamma_H(t)+\Gamma_V(t)]} \varrho_{VH}(0) & e^{-\Gamma_V(t)} \varrho_{VV}(0) \\
\end{array}%
\right) ,
\end{equation*}

\noindent where $\Gamma_H(t) = \int_0^t \gamma_H(t') dt'$ and
$\Gamma_V(t) = \int_0^t \gamma_V(t') dt'$. If $\gamma_H(t)$ and
$\gamma_V(t)$ are both nonnegative, then the trace decreasing
dynamical map $\Lambda(t)$ is CP-divisible.

Let us consider the specific rates
\begin{eqnarray}
&& \gamma_H(t) = \gamma \left( 1 - \frac{\omega \cos\omega
t}{\sqrt{\gamma^2 + \omega^2} + \gamma \sin\omega t} \right), \label{gamma-H-t}\\
&& \gamma_V(t) = \gamma \left( 1 + \frac{\omega \cos\omega
t}{\sqrt{\gamma^2 + \omega^2} - \gamma \sin\omega t} \right),
\label{gamma-V-t}
\end{eqnarray}

\noindent such that $\gamma_H(t) \geq 0$ and $\gamma_V(t) \geq 0$
if $\gamma > 0$, $\omega > 0$. Then
\begin{eqnarray}
&& p_H(t):= e^{-\Gamma_H(t)} = e^{- \gamma t} \left( 1 + \frac{\gamma}{\sqrt{\gamma^2 + \omega^2}} \sin\omega t \right), \label{pH-t}\\
&& p_V(t):= e^{-\Gamma_V(t)} = e^{- \gamma t} \left( 1 -
\frac{\gamma}{\sqrt{\gamma^2 + \omega^2}} \sin\omega t \right)
\label{pV-t}
\end{eqnarray}

\noindent are both nonincreasing functions of time $t$. The
physical meaning of $p_H(t)$ is the probability to successfully
detect the horizontally polarized photon at time $t$ provided the
photon was initially horizontally polarized. Eq.~\eqref{pV-t} has
a similar meaning for vertically polarized photons.
CP-divisibility of $\Lambda(t)$ means that the time evolution can
be represented as a sequence of independent subevolutions, with
each of them being a valid trace decreasing quantum operation. In
this example, all the subevolutions are polarization dependent
losses with different attenuation factors.

Let $\ket{\psi_+}\bra{\psi_+}$ be a maximally entangled initial
state of the qubit system and a two-dimensional ancilla. This is
exactly the case in the experimental scenario of
Ref.~\cite{mataloni-2019}, where the system-ancilla polarization
state of two photons is produced via the spontaneous parametric
down-conversion process. To be precise, $\ket{\psi_+} =
\frac{1}{\sqrt{2}}(\ket{H} \otimes \ket{H} + \ket{V} \otimes
\ket{V})$. The system-ancilla trace decreasing and CP-divisible
dynamics $\Lambda(t) \otimes {\rm Id} [\ket{\psi_+}\bra{\psi_+}]$
results in the following conditional dynamics:
\begin{eqnarray}
&& \frac{\Lambda(t) \otimes {\rm Id}
[\ket{\psi_+}\bra{\psi_+}]}{{\rm tr}\big[ \Lambda(t) \otimes {\rm
Id} [\ket{\psi_+}\bra{\psi_+}] \big]} =
\ket{\varphi(t)}\bra{\varphi(t)}, \\
&& \ket{\varphi(t)} = \frac{1}{\sqrt{p_H(t) + p_V(t)}} \left(
\sqrt{p_H(t)} \ket{H} \otimes \ket{H} + \sqrt{p_V(t)} \ket{V}
\otimes \ket{V} \right).
\end{eqnarray}

\noindent Note that the conditional system-ancilla state remains
pure for all times $t \geq 0$, so its entanglement is readily
quantified by an entanglement monotone called
concurrence~\cite{hill-wootters-1997}. We use that quantifier for
our results to be comparable with Ref.~\cite{mataloni-2019}, where
the dynamics of concurrence is analyzed by an experimental
reconstruction of the conditional system-ancilla states at
different time moments. In our example, we have
\begin{equation} \label{E-t-pdl}
{\cal E}(t) = \frac{2 \sqrt{p_H(t) p_V(t)}}{p_H(t) + p_V(t)} =
\sqrt{\frac{\gamma^2 \cos^2\omega t + \omega^2}{\gamma^2 +
\omega^2}},
\end{equation}

\noindent which is clearly a nonmonotonic function of time if
$\gamma > 0$, $\omega > 0$. Note that $\max_{t\geq 0} {\cal E}(t)
= {\cal E}(\frac{\pi n}{\omega}) = 1$ and $\min_{t\geq 0} {\cal
E}(t) = {\cal E}(\frac{\pi}{2\omega} + \frac{\pi n}{\omega}) =
\frac{\omega}{\sqrt{\gamma^2 + \omega^2}}$, $n \in \mathbb{N}$, so
there are infinitely many time periods of the entanglement
increase (see the upper solid line in Fig.~\ref{figure-3}). The
minimum value $\min_{t\geq 0} {\cal E}(t)$ can be arbitrarily
close to 0 if $\omega \ll \gamma$. \hfill $\triangle$
\end{example}

\begin{figure}
\begin{center}
\includegraphics[width=10cm]{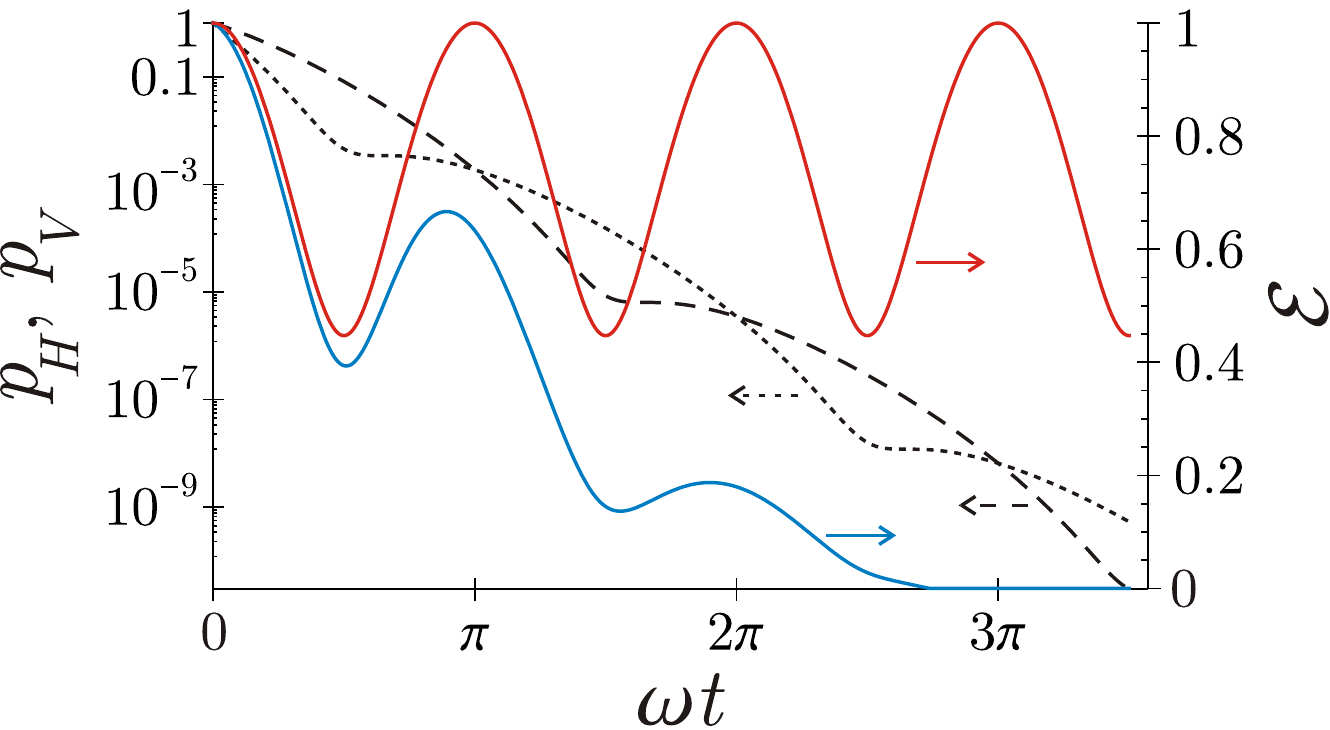}
\end{center}
\caption{Monotonic dynamics of the probability $p_H$ ({\it dashed
line}) and the probability $p_V$ ({\it dotted line}) in the trace
decreasing CP-divisible evolution governed by the master
equation~\eqref{master-pdl} for polarization dependent losses; the
rates $\gamma_H(t)$ and $\gamma_V(t)$ are given by
Eq.~\eqref{gamma-H-t} and Eq.~\eqref{gamma-V-t}, respectively,
where $\gamma = 2\omega$. Nonmonotonic dynamics of the concurrence
${\cal E}$ for the normalized system-ancilla state in
Example~\ref{example-pdl-t} with  $\gamma = 2\omega$ ({\it upper
solid line}), Eq.~\eqref{E-t-pdl}. Nonmonotonic dynamics of the
concurrence ${\cal E}$ for the normalized system-ancilla state in
Example~\ref{example-pdl-t-noise} with $\gamma = 2 \omega$ and the
depolarization rate $\lambda = 0.05 \omega$ ({\it lower solid
line}).} \label{figure-3}
\end{figure}

Therefore, the entanglement increase in the experimentally
reconstructed postselected system-ancilla state cannot be a valid
indicator of CP-indivisibility (non-Markovianity) for trace
decreasing dynamical maps. This fact was overlooked in
Ref.~\cite{mataloni-2019}.

The dynamical map $\Lambda(t)$ considered in
Example~\ref{example-pdl-t} has Kraus rank 1, which preserves the
purity of the postselected state. In the following example, we
consider a more sophisticated dynamics, where the conditional
system-ancilla state becomes mixed.

\begin{example} \label{example-pdl-t-noise}
Suppose the polarization dependent losses in
Example~\ref{example-pdl-t} are accompanied by a standard
depolarization with the rate $\lambda > 0$. Then the master
equation takes the form
\begin{equation} \label{master-pdl-depolarization}
\frac{d \varrho(t)}{dt} = - \frac{1}{2} \Big\{ \gamma_H(t)
\ket{H}\bra{H} + \gamma_V(t) \ket{V}\bra{V},\varrho(t) \Big\} +
\frac{\lambda}{4}\sum_{i=1}^3 \big( \sigma_i \varrho(t) \sigma_i -
\varrho(t) \big) ,
\end{equation}

\noindent where $\gamma_H(t)$ and $\gamma_V(t)$ are given by
Eqs.~\eqref{gamma-H-t} and \eqref{gamma-V-t}, respectively;
$(\sigma_1,\sigma_2,\sigma_3)$ is the conventional set of Pauli
operators, i.e., $\sigma_1 = \ket{H}\bra{V} + \ket{V}\bra{H}$,
$\sigma_2 = -i \ket{H}\bra{V} + i \ket{V}\bra{H}$, $\sigma_3 =
\ket{H}\bra{H} - \ket{V}\bra{V}$.
Eq.~\eqref{master-pdl-depolarization} defines a trace decreasing
CP-divisible dynamical map $\Lambda(t)$. The analytical expression
for the postselected state $\Lambda(t) \otimes {\rm Id}
[\ket{\psi_+}\bra{\psi_+}] / {\rm tr}\big[ \Lambda(t) \otimes {\rm
Id} [\ket{\psi_+}\bra{\psi_+}] \big]$ is rather involved though
straightforward. We depict the postselected state concurrence
${\cal E}(t)$ for the particular choice of parameters $\gamma$,
$\omega$, and $\lambda$ in Fig.~\ref{figure-3} (see the lower
solid line). The concurrence ${\cal E}(t)$ is a nonmonotonic
function of time $t$ despite the fact that $\Lambda(t)$ is
CP-divisible. \hfill $\triangle$
\end{example}

We can make another important observation in Fig.~\ref{figure-3}:
If the depolarization rate is strictly positive in
Example~\ref{example-pdl-t-noise} (i.e., $\lambda>0$), then the
concurrence ${\cal E}(t)$ vanishes at some time moment and remains
zero after. In what follows we prove that this is a general
phenomenon.

\begin{proposition}
Let $\Lambda(t)$ be a trace decreasing CP-divisible dynamical map
and $\varrho_{SA}(0)$ be an initial system-ancilla state. If the
postselected system-ancilla state $\Lambda(t_{\ast}) \otimes {\rm
Id} [\varrho_{SA}(0)] / {\rm tr}\big[ \Lambda(t_{\ast}) \otimes
{\rm Id} [\varrho_{SA}(0)] \big]$ becomes separable at time moment
$t_{\ast}$, then all future postselected states ($t \geq
t_{\ast}$) remain separable.
\end{proposition}
\begin{proof}
Since the postselected system-ancilla state is separable at time
$t_{\ast}$, we have $\Lambda(t_{\ast}) \otimes {\rm Id}
[\varrho_{SA}(0)] / {\rm tr}\big[ \Lambda(t_{\ast}) \otimes {\rm
Id} [\varrho_{SA}(0)] \big] = \sum_i p_i \varrho_S^{(i)} \otimes
\varrho_A^{(i)}$, where $\{p_i\}_i$ is a probability distribution,
$\{\varrho_S^{(i)}\}_i$ and $\{\varrho_A^{(i)}\}_i$ are the sets
of density operators for the system and the ancilla, respectively.
Let $t \geq t_{\ast}$. Since $\Lambda(t)$ is CP-divisible, then
there exists a map $\Theta(t,t_{\ast})$ such that $\Lambda(t) =
\Theta(t,t_{\ast}) \Lambda(t_{\ast})$. Hence,
\begin{eqnarray*}
&& \frac{\Lambda(t) \otimes {\rm Id} [\varrho_{SA}(0)]}{{\rm
tr}\big[ \Lambda(t) \otimes {\rm Id} [\varrho_{SA}(0)] \big]} =
\sum_{i} q_i \frac{\Theta(t,t_{\ast})[\varrho_S^{(i)}]}{{\rm tr}
\left[ \Theta(t,t_{\ast})[\varrho_S^{(i)}] \right]} \otimes
\varrho_A^{(i)}, \\
&& q_i = \frac{p_i {\rm tr} \left[
\Theta(t,t_{\ast})[\varrho_S^{(i)}] \right]}{\sum_j p_j {\rm tr}
\left[ \Theta(t,t_{\ast})[\varrho_S^{(j)}] \right]},
\end{eqnarray*}

\noindent i.e., the postselected state at time moment $t \geq
t_{\ast}$ is separable too. \hfill $\square$
\end{proof}

\begin{corollary} \label{corollary-1}
Suppose the system experiences a trace decreasing dynamics
$\Lambda(t)$. The system-ancilla entanglement death followed by
the entanglement revival is an indication of CP-indivisibility for
$\Lambda(t)$.
\end{corollary}

Using the results of Ref.~\cite{sperling-2011}, the above claims
can be reformulated and generalized in terms of the Schmidt rank
as follows: A CP-divisible trace decreasing system dynamics
$\Lambda(t)$ does not increase the Schmidt rank of the
system-environment normalized states $\Lambda(t) \otimes {\rm Id}
[\varrho_{SA}(0)] / {\rm tr}\big[ \Lambda(t) \otimes {\rm Id}
[\varrho_{SA}(0)] \big]$.

Clearly, the examples and the statements above hold true not only
for continuous time evolutions but also in the case of the
discrete time evolutions, e.g., collision models that are reviewed
in Ref.~\cite{campbell-2021}.

\section{Generalized erasure dynamics}

In Ref.~\cite{filippov-2021}, a generalized erasure channel has
been introduced as a physically motivated extension of a trace
decreasing operation. In fact, the successful implementation of a
quantum operation $\Lambda$ for an input state $\varrho$ happens
with the probability ${\rm tr}\big[ \Lambda[\varrho] \big]$, so
the failure probability equals
\begin{equation}
p_{\rm fail} = 1 - {\rm tr}\big[ \Lambda[\varrho] \big] = {\rm
tr}\big[ \varrho - \Lambda[\varrho] \big] = {\rm tr}\left[ \left(
I - \Lambda^{\dag}[I] \right) \varrho \right],
\end{equation}

\noindent where $\Lambda^{\dag}$ is a dual map with respect to
$\Lambda$ such that ${\rm tr}\big[ \Lambda[X] Y \big] = {\rm
tr}\big[ X \Lambda^{\dag}[Y] \big]$ for all trace-class operators
$X$ and bounded operators $Y$. The probability $p_{\rm fail}$ is
readily accessible in a physical experiment. For instance, in the
experimental setup of Ref.~\cite{mataloni-2019}, this probability
can be evaluated as $(N_A - N_S)/ N_S$, where $N_A$ is a number of
the detector clicks for the ancilla and $N_S$ is a number of the
detector clicks for the system, see Fig.~\ref{figure-2}(b).
Although this extra information is free, it is usually disregarded
in the postselection-oriented experiments. We argue that the
unsuccessful implementation of operation $\Lambda$ can be viewed
as an erasure event. In Fig.~\ref{figure-2}(b), this event is
associated with the click of the ancilla detector and no click of
the system detector. We extend the Hilbert space and include the
erasure flag state $\ket{e}$ into it so that the map
\begin{equation}
\Gamma_{\Lambda}[\varrho] = \Lambda[\varrho] \oplus {\rm tr}\left[
\left( I - \Lambda^{\dag}[I] \right) \varrho \right] \
\ket{e}\bra{e} \equiv \left(%
\begin{array}{cc}
  \Lambda[\varrho] & {\bf 0} \\
  {\bf 0}^{\top} & {\rm tr}\left[
\left( I - \Lambda^{\dag}[I] \right) \varrho \right] \\
\end{array}%
\right)
\end{equation}

\noindent is trace preserving and completely
positive~\cite{filippov-2021}. Therefore, any trace decreasing
dynamics $\Lambda(t)$ can be alternatively described by virtue of
the corresponding trace preserving dynamical map
$\Gamma_{\Lambda}(t) \equiv \Gamma_{\Lambda(t)}$. We refer to
$\Gamma_{\Lambda}(t)$ as the generalized erasure dynamics.

CP-indivisibility identification problems discussed in
Secs.~\ref{section-distinguish} and \ref{section-s-a-ent} were
caused by the postselection. Instead, if we use the generalized
erasure dynamics, then those problems are automatically resolved.

\begin{proposition} \label{prop-CPdiv-GEC}
Suppose $\Lambda(t)$ is a CP-divisible trace decreasing dynamical
map, then $\Gamma_{\Lambda}(t)$ is a CP-divisible trace preserving
dynamical map.
\end{proposition}

\begin{proof}
Let $t_2 \geq t_1 \geq 0$, then $\Lambda(t_2) = \Theta(t_2,t_1)
\Lambda(t_1)$, where $\Theta(t_2,t_1)$ is a quantum operation.
Consider the map
\begin{equation} \label{Xi-channel}
\Xi(t_2,t_1) \left[ \left(%
\begin{array}{cc}
  \varrho & \vdots \\
  \cdots & c \\
\end{array}%
\right) \right] = \left(%
\begin{array}{cc}
  \Theta(t_2,t_1) [\varrho] & {\bf 0} \\
  {\bf 0}^\top & c + {\rm tr}\left[\rho (I - \Theta(t_2,t_1)^{\dag}[I]) \right] \\
\end{array}%
\right),
\end{equation}

\noindent which is automatically trace preserving and, in
addition, completely positive because a quantum operation
$\Theta(t_2,t_1)$ is trace nonincreasing, i.e.,
$\Theta(t_2,t_1)^{\dag}[I] \leq I$. We have $\Gamma_{\Lambda}(t_2)
= \Xi(t_2,t_1) \Gamma_{\Lambda}(t_1)$, i.e., $\Gamma_{\Lambda}(t)$
is CP-divisible. \hfill $\square$
\end{proof}

Due to Proposition~\ref{prop-CPdiv-GEC}, the distinguishability of
equiprobable states $\Gamma_{\Lambda}(t)[\varrho_1]$ and
$\Gamma_{\Lambda}(t)[\varrho_2]$ monotonically decreases if
$\Lambda(t)$ is CP-divisible. Remarkably, the probability to
successfully distinguish these states, $\frac{1}{2} + \frac{1}{4}
\| \Gamma_{\Lambda}(t) [\varrho_1] - \Gamma_{\Lambda}(t)
[\varrho_2] \|_1$, exceeds the probability~\eqref{success-correct}
by $\frac{1}{2}\left( 1 - \min_{i=1,2} {\rm tr} \big[
\Lambda(t)[\varrho_{i}] \big] \right)$ because the modified
scenario takes into account the erasure events too [see
Fig.~\ref{figure-1}(b), where $\nu$ is a coarse-graining
postprocessing with two outcome events: $k=k_{\ast}$ and $k \neq
k_{\ast}$].

Similarly, Proposition~\ref{prop-CPdiv-GEC} implies that the
system-ancilla entanglement monotonically decreases during the
trace preserving dynamics $\Gamma_{\Lambda}(t) \otimes {\rm Id}$
if $\Lambda(t)$ is CP-divisible.

\section{Conclusions}

In the present work, we have paid attention to some simple yet
important questions of how to analyze the divisibility problem in
experiments that involve postselection. By a number of examples we
have demonstrated the complete inconsistency of some approaches
that operate with the postselected states as though they were
obtained as a result of a trace preserving dynamics. In
particular, none of the Eqs.~\eqref{td-naive} and
\eqref{td-naive-2} is adequate in the analysis of the information
backflow because either of these quantities can increase during a
CP-divisible dynamics. The original idea behind the information
backflow is properly described by Eq.~\eqref{success-correct} that
has a clear physical meaning. Similarly, we show that the increase
in the system-ancilla entanglement is not a valid indicator of
CP-indivisibility if the postselection takes place. Unfortunately,
these facts are sometimes overlooked in research papers, e.g., a
nonmonotonic concurrence behavior for the system-ancilla
postselected state was mistreated as an indicator of
non-Markovianity in the paper~\cite{mataloni-2019} (this paper has
an extra drawback related with non-unitarity of the
inter-environment collision that we do not discuss in this work).
We advocate that a correct indicator of non-Markovianity should be
related with an increase of the system-ancilla Schmidt rank, e.g.,
the revival of the entanglement after the entanglement death
(Corollary~\ref{corollary-1}). A general analysis of the two-qubit
entanglement dynamics under local trace decreasing maps can be
accomplished by following the lines of Ref.~\cite{ffk-2018} and
using the quantum Sinkhorn theorem. Finally, we have described the
framework of generalized erasure dynamics that takes into account
how often a desired trace decreasing operation is actually
implemented and how often the implementation fails. This extra
information is readily available in experiments though it is
rarely used. On the other hand, including this information into
the description allows us to return to the fold of trace
preserving maps. Those trace preserving maps inherit some
properties of the original operations, e.g., CP-divisibility
(Proposition~\ref{prop-CPdiv-GEC}). Moreover, the induced trace
preserving maps have an interesting property of superadditivity of
coherent information~\cite{filippov-2021}.

%
%

\end{document}